\documentclass[12pt]{article}

\usepackage{authblk}
\usepackage{amsmath}
\usepackage{amssymb}
\usepackage{enumerate}
\usepackage{framed}

\newtheorem{theorem}{\bf Theorem}[section]

\newtheorem{definition} [theorem] {\bf Definition}

\newenvironment{proof}{\noindent\mbox{\textbf{Proof} 
}}{\rm\hspace*{\fill}$\rule{7pt}{7pt}$\vspace{10pt}}

\title{On the~set-representable orthomodular posets that are point-distinguishing}

\author{Dominika Bure\v{s}ov\'{a} and Pavel Pt\'{a}k}

\date{}

\begin{document}

\maketitle

\begin{abstract}
 Let us denote by $\mathcal{SOMP}$ the~class of all set-representable orthomodular posets and by $\mathcal{PD SOMP}$ those elements of $\mathcal{SOMP}$ in which any pair of points in the~underlying set $P$ can be distinguished by a~set (i.e., $(P, \mathcal{L}) \in \mathcal{PD SOMP}$ precisely when for any pair $x, y \in P$ there is a set $A \in \mathcal{L}$ with $x \in A$ and $y \notin A$).
  In this note we first construct, for each $(P, \mathcal{L}) \in \mathcal{SOMP}$, a~point-distinguishing orthomodular poset that is isomorphic to $(P, \mathcal{L})$. We show that by using a~generalized form of the~Stone representation technique we also obtain point-distinguishing representations  of $(P, \mathcal{L})$. We then prove that this technique gives us point-distinguishing representations on which all two-valued states are determined by points (all two-valued states are Dirac states).
   Since orthomodular posets may be regarded as abstract counterparts of event structures about quantum experiments, results of this work may have some relevance for the~foundation of quantum mechanics.
\end{abstract}



\section{Introduction}
In the~logical-algebraic foundation of quantum mechanics, orthomodular posets are viewed as event structures about quantum experiments. They are called quantum logics. Among those orthomodular posets, a special conceptual position is taken by the~set-representable ones (see~\cite{BuresovaPtak, DeSimoneNavaraPtak1, DeSimoneNavaraPtak2, DvurecenskijPulmannova, GoGre, Gudder, Harding, Ptak} etc.). In this note, \textbf{we shall exclusively deal with set-representable orthomodular posets}. Let us introduce them.

\begin{definition}
    Let $P$ be a set and let $\mathcal{L}$ be a collection of subsets of $P$ ($\mathcal{L}\subseteq \mathrm{exp}~P$). Suppose that $\mathcal{L}$ is subject to the~following requirements:
    \begin{enumerate}[(i)]
        \item $P \in \mathcal{L}$,
        \item if $A\in\mathcal{L}$, then $A'\in\mathcal{L}$ ($A' = P\setminus A$),
        \item if $A,B\in\mathcal{L}$ and $A\cap B=\emptyset$, then $A\cup B\in\mathcal{L}$.
    \end{enumerate}
    Then $(P, \mathcal{L})$ is said to be a set-representable orthomodular poset ($(P, \mathcal{L}) \in \mathcal{SOMP}$). If $\mathcal{L}$ is a lattice with respect to the~inclusion ordering, the~couple ($P, \mathcal{L}$) is said to be a set-representable orthomodular lattice.
\end{definition}
 	%   
    %    
    %    
    %   
    %    SECTION 2
    %    
    %    

\section{Point-distinguishing orthomodular posets}
The~basic notion of this paragraph is introduced by the~next definition.
\begin{definition}
    Let $(P, \mathcal{L})$ belong to $\mathcal{SOMP}$.~Let us call $(P, \mathcal{L})$\\ \textbf{point-distinguishing} provided for each couple of distinct elements $x,y \in P$ there is a set $ A \in \mathcal{L}$ such that $x \in A$ and $y \in P\setminus A$.
\end{definition}
    Obviously, many orthomodular posets of $\mathcal{SOMP}$ are not point-distinguishing (for instance, such is ($P, \mathcal{L}) = (\emptyset, A, A', B, B', P)$ ``for big sets'' $A$ and $B$) and so are many Boolean algebras. The~objective of this part of the~paper, seemingly not explicitly dealt with in the~investigation of $\mathcal{SOMP}$, is to show that each $(P, \mathcal{L})$ has a~point-distinguishing representation. We will first provide an~internal construction, then we will apply a procedure - a generalized form of the~Stone representation theorem - borrowed from the~theory of Boolean algebras.
    
Let us first introduce a natural relation on the~orthomodular posets of $\mathcal{SOMP}$. 
\begin{definition}
Let $(P, \mathcal{L})$ belong to $\mathcal{SOMP}$. Let $\mathcal{R}$ be the~relation on $P$ defined as follows:
If $x,y \in P$, then $x\mathcal{R}y$ if there is no set $A\in \mathcal{L}$ such that $x \in A$ and $y \in A'$.
\end{definition}
\begin{theorem}
Let $(P, \mathcal{L})$ belong to $\mathcal{SOMP}$. Then the~relation $\mathcal{R}$ on $P$ defined in \textit{Definition 2.2} is an~equivalence relation. Moreover, if $\{R_{\alpha}~|~\alpha\in I\}$ is the~decomposition of $P$ by the~classes of the~equivalence
$\mathcal{R}$, then the~following implication holds true:

If $A \in \mathcal{L}$, then, for any $\alpha \in I$, either $R_{\alpha}\subseteq A$, or $R_{\alpha}\subseteq A'$.
\end{theorem}
\begin{proof}
The~relation $\mathcal{R}$ is obviously reflexive and symmetric. As regards the~transitivity, let us suppose that $x\mathcal{R}y$ and $y\mathcal{R}z$ ($x,y,z\in P$). Then $x\mathcal{R}z$. Indeed, if there is a set $A\in\mathcal{L}$ such that $x\in A$ and $z\in A'$, then either $y\in A'$ in which case $x$(non~$\mathcal{R}$)$y$, or $y\in A$ in which case $y$(non~$\mathcal{R}$)$z$. Further, if $A\in\mathcal{L}$, then $R_{\alpha}$ cannot intersect both $A$ and $A'$ since if it does, $R_{\alpha}$ would not be a~class of $\mathcal{R}$.
\end{proof}
\\
The~previous theorem allows us to formulate the~following result. Prior to that, let us recall a~notion and make some conventions. Suppose that $(P, \mathcal{L})$ and $(Q, \mathcal{K})$ belong to $\mathcal{SOMP}$. Let us call a mapping $f~:~\mathcal{L}\to\mathcal{K}$ a~$\mathcal{SOMP}$-morphism if
\begin{enumerate}[(i)]
        \item $f(P) = Q$,
        \item $f(A') = f(A)'~(A \in \mathcal{L})$, and
        \item $f(A \cup B) = f(A) \cup f(B)$ provided $A \subseteq B'$ $(A, B \in \mathcal{L})$.
    \end{enumerate}
A both injective and surjective $\mathcal{SOMP}$-morphism is said to be a~\textbf{ $\mathcal{SOMP}$-isomorphism} if $f^{-1}$ is a~$\mathcal{SOMP}$-morphism, too.
\begin{definition}
Let $(P, \mathcal{L})$ belong to $\mathcal{SOMP}$. Let $\mathcal{R}$ be the~equivalence introduced in \textit{Definition 2.2} and let $\{R_{\alpha}~|~\alpha\in I\}$ be the~decomposition of $P$ given by the~classes of $\mathcal{R}$. Consider the~following couple $(\tilde{P},~\tilde{\mathcal{L}})$, where $\tilde{P} = \{R_{\alpha}~|~\alpha\in I\}$ and $\tilde{\mathcal{L}}$ is the~following collection of subsets of $\tilde{P}: \tilde{A}\in\tilde{\mathcal{L}}$ if there is a set $A\in\mathcal{L}$ such that $\tilde{A}=\{R_{\alpha} \mid R_{\alpha}\subseteq A\}$. Let us call $(\tilde{P},~\tilde{\mathcal{L}})$ the~\textbf{natural point-distinguishing representation} of $(P, \mathcal{L})$.
\end{definition}
Obviously, $(\tilde{P},~\tilde{\mathcal{L}})$ is point-distinguishing: If $R_{\alpha} \neq R_{\beta}$ and $x \in R_{\alpha}$ and $y \in R_{\beta}$ then there is $A\in \mathcal{L}$ such that $x \in A$ and $y \in A'$. By \textit{Theorem 2.3} $R_{\alpha}\subseteq A$ and $R_{\beta}\subseteq A'$. In order to distinguish $R_{\alpha}$ from $R_{\beta}$ it is sufficient to take $\tilde{A} = \{A_{\gamma}~|~A_{\gamma}\subseteq A\}$.

    %   
    %    
    %    
    %   Theorem 2.5
    %    
    %    
    %   
\begin{theorem}
Let $(P, \mathcal{L})$ belong to $\mathcal{SOMP}$ and let $(\tilde{P}, \tilde{\mathcal{L}})$ be the~natural point-distinguishing representation of $(P, \mathcal{L})$.
\begin{enumerate}[(i)]
        \item If $f:\mathcal{L}\to \tilde{\mathcal{L}}$ assigns to any $A\in\mathcal{L}$ the~set $f(A) = \{R_{\alpha}~|~R_{\alpha}\subseteq A \} \in \tilde{\mathcal{L}}$, then $f$ is a $\mathcal{SOMP}$-isomorphism.
        \item If $\mathcal{L}$ belongs to $\mathcal{SOMP}$ and $\mathcal{L}$ is a lattice, then $f:\mathcal{L}\to \tilde{\mathcal{L}}$ defined above is a lattice $\mathcal{SOMP}$-isomorphism.
        \item If $(P, \mathcal{L})$ is closed under the~formation of symmetric difference, then so is $(\tilde{P}, \tilde{\mathcal{L}})$ and both $f$ and $f^{-1}$ preserve the~respective symmetric differences.
        \item If $(P, \mathcal{L})$ is a Boolean algebra then so is $(\tilde{P}, \tilde{\mathcal{L}})$ and $f$ is a Boolean isomorphism.
    \end{enumerate}
\end{theorem}

\begin{proof}
It easily follows from \textit{Theorem 2.3}.
\end{proof}\\
It could be noted that \textit{Theorem 2.5(ii)} may shed light on the~algebraic theory of the~subclass of $\mathcal{SOMP}$ consisting of lattices- this subclass forms a variety (see~\cite{GoGre, Mayet, MayetPtak}; obviously, the~point-distinguishing lattices of $\mathcal{SOMP}$ are easier to deal with). The~statement of \textit{Theorem 2.5(iii)} could be instrumental in the~study of the~orthomodular posets of $\mathcal{SOMP}$ that have a symmetric difference (see e.g.~\cite{BuresovaPtak, DeSimoneNavaraPtak2, Matousek}). In particular, it is worth noticing that the~examples of conceptually important orthomodular posets with the~property that $x$ is compatible with $y$ exactly when $x \lor y$ exists (see~\cite{NavaraPtak}) could be constructed point-distinguishing.

For a potential further research, let us conclude this paragraph with a~few observations concerning the~natural point-distinguishing representation.
\begin{enumerate}[(i)]
        \item The~natural point-distinguishing representation is functorial. Indeed, if $(P, \mathcal{L})$ and $(Q, \mathcal{K})$ belong to $\mathcal{SOMP}$ and $\tilde{f}: (P, \mathcal{L}) \to (\tilde{P}, \tilde{\mathcal{L}})$ and $\tilde{\tilde{f}}: (Q, \mathcal{K}) \to (\tilde{Q}, \tilde{\mathcal{K}})$ are the~natural point-distinguishing representations, then for each $\mathcal{SOMP}$-morphism $g: (P, \mathcal{L}) \to (Q, \mathcal{K})$ there is a unique $\mathcal{SOMP}$-morphism $\tilde{g}: (\tilde{P}, \tilde{\mathcal{L}}) \to (\tilde{Q}, \tilde{\mathcal{K}})$ such that we have a commuting diagram $\tilde{\tilde{f}}*g = \tilde{g}*\tilde{f}$.
        \item The~natural point-distinguishing representation of $(P, \mathcal{L})$ can be considered ``to live  on a subset of $P$''. Indeed, we can choose a~point, $x_{\alpha}$, in any $R_{\alpha}$ and ``copy''  $\tilde{\mathcal{L}}$ on the~set  $Q = \{x_{\alpha}\mid \alpha \in I\}$, i.e. we include a subset $\{x_{\alpha}\mid \alpha \in J\}$ of $Q$ exactly when $\{R_{\alpha} \mid \alpha \in J \} \in \mathcal{\tilde{L}}$.
        \item If $(Q, \mathcal{K})$ is $\mathcal{SOMP}$-isomorphic to $(P, \mathcal{L})$ and $(Q, \mathcal{K})$ is point-distinguishing, then $\mathrm{card}~\tilde{\mathcal{L}} \leq \mathrm{card}~Q$.
        \item Let $S$ be a set and let $\{S_\alpha~|~\alpha \in I\}$be a partition of $S$. Let us denote by $(S, \mathcal{B}_{\min})$ (resp. $(S, \mathcal{B}_{\max})$) the~Boolean algebra of the~finite and co-finite unions of the~sets $\{S_\alpha~|~\alpha \in I\}$ (resp. the~Boolean algebra of all subsets of $\{S_\alpha~|~\alpha \in I\}$). Then $\mathcal{B}_{\min}$ (resp. $\mathcal{B}_{\max}$) is a minimal (resp.maximal) Boolean algebra of subsets of $S$ for which $\{S_{\alpha}, \alpha \in I\}$ forms the~classes of the~natural point-distinguishing equivalence.
    \end{enumerate}
    %   
    %    
    %    
    %   CHAP 3
    %    
    %    
    %   
\section{Constructing point-distinguishing orthomodular posets by the~Stone technique}
Let $(P, \mathcal{L})$ belong to $\mathcal{SOMP}$. Let us denote by $\mathcal{S}_2(P, \mathcal{L})$ the~set of all two-valued states over $(P, \mathcal{L})$. Recall that $s: \mathcal{L} \to \{0,1\}$ is said to be a~\textbf{two-valued state} over $(P, \mathcal{L})$ if
\begin{enumerate}[(i)]
\item $s(P) = 1$, and
\item $s(A \cup B) = s(A) + s(B)$ provided $A, B \in \mathcal{L}$ and $A \subseteq B'$.
\end{enumerate}
A two-valued state $s$ over $(P, \mathcal{L})$ is said to be a~\textbf{Dirac state} if there is a~point $p \in P$ such that $s(A) = 1$ exactly when $p \in A$.\\

The~following result generalizes the~Boolean Stone representation theorem. We style it for our purpose.
\begin{theorem}
Let $(P, \mathcal{L})$ belong to $\mathcal{SOMP}$. Let $\mathcal{S} \subseteq \mathcal{S}_2(P, \mathcal{L})$ and let $\mathcal{S}$ have the~following property:

If $A, B \in \mathcal{L}$ and $A \nsubseteq B$, then there is $s \in \mathcal{S}$ with $s(A) = 1$ and $s(B) = 0$. Let $\mathcal{U}$ be the~collection of the~subsets $U$ of $\mathcal{S}$ determined as follows:

$U \in \mathcal{U}$ exactly when there is a set $Q \in \mathcal{L}$ such that $U = \{s \in \mathcal{S}~|~s(Q) = 1 \}$. Then $(\mathcal{S}, \mathcal{U}) \in \mathcal{SOMP}$ and $(\mathcal{S}, \mathcal{U})$ is point-distinguishing. Moreover, $(\mathcal{S}, \mathcal{U})$ is $\mathcal{SOMP}$-isomorphic to $(P, \mathcal{L})$. If $\mathcal{S}$ consists of all Dirac states, then 
$(\mathcal{S}, \mathcal{U})$ is $\mathcal{SOMP}$-isomorphic to the~natural point-distinguishing representation $(\tilde{P}, \tilde{\mathcal{L}})$ of $(P, \mathcal{L})$. If $\mathcal{S} = \mathcal{S}_2(P, \mathcal{L})$, then each two-valued state on $(\mathcal{S}, \mathcal{U})$ is a Dirac state.
\end{theorem}

\begin{proof}
If we define a mapping $e: \mathcal{L} \to \mathcal{U}$ by setting $e(C) = \{s \in \mathcal{S}~|~s(C) = 1 \}, C \in \mathcal{L} $, we can easily prove-verbatim the~Stone representation theorem- that $e$ is a $\mathcal{SOMP}$-isomorphism (see also~\cite{Tkadlec} for related considerations). Moreover, if $s_1 \neq s_2$ then there is a set $D$, $D \in \mathcal{L}$ such that $s_1(D) = 1$ and $s_2(D) = 0$. If we take $U_1 = \{s \in \mathcal{S}~|~s(D) = 1\}$, then $U_1 \in \mathcal{U}$ and $s_1 \in U_1$ whereas $s_2 \notin U_1$. Further, suppose that $\mathcal{S} = \{s \in \mathcal{S}_2(P, \mathcal{L})~|~s$ is a~Dirac state\}. If $s_1 \in \mathcal{S}$ and $s_2 \in \mathcal{S}$ with $s_1$ and $s_2$ given by $p_1 \in P$ and $p_2 \in P$, we see that $s_1 = s_2$ precisely when $p_1\mathcal{R}p_2$ in the~equivalence $\mathcal{R}$ defined in \textit{Definition 2.2}. Hence in this case the~orthomodular poset $(\mathcal{S}, \mathcal{U})$ is $\mathcal{SOMP}$-isomorphic to the~natural point-distinguishing representation $(\tilde{P},\tilde{\mathcal{L}})$ of $(P, \mathcal{L})$. Finally, suppose that $\mathcal{S} = \mathcal{S}_2(P, \mathcal{L})$. Let $t$ be a two-valued state on $(\mathcal{S}, \mathcal{U})$. Applying the~$\mathcal{SOMP}$-isomorphism $e: \mathcal{L} \to \mathcal{U}$, let us consider the~two-valued state, $s$, on $(P, \mathcal{L})$ such that $s = te$. So $s \in \mathcal{S}_2(P, \mathcal{L})$ and we conclude that $t$ is the~Dirac state given by $s$. Indeed, suppose that $t (U) = 1$ for some $U \in \mathcal{U}$. Write $U = \{t~|~t \in \mathcal{S}_2(P, \mathcal{L})$ such that $u(V) = 1$ for some $V \in \mathcal{L}\}$. Then $s(V) = te(V) = 1$ and we conclude that $t$ is the~Dirac state given by $s$. The~proof is complete.
\end{proof}\\

Let us shortly comment on the~results obtained. First, for the~property of having all two-valued states Dirac states one has to pay the~price of having to lift up the cardinality of the~underlying set. In principal, the~cardinality of $\mathcal{S}_2(P, \mathcal{L})$ might be ${2^{2}}^{\mathrm{card}P}$. In the~case of $P$ being finite, the~set $\mathcal{S}_2(P, \mathcal{L})$ is finite, too, and this can be used in the~analysis of states of finite orthomodular posets (see e.g.~\cite{BikchentaevNavara}).

Another question concerns the~algebraic structure of those elements of $\mathcal{SOMP}$ on which all two-valued states are Dirac states. This class is closed under finite products. On the~other hand, this class is not closed under substructures. Indeed, if we e.g. take for $(P, \mathcal{L})$ the~ortomodular poset $6_{even}$ of all subsets of $\{1,2,3,4,5,6\}$ with an~even cardinality, then each two-valued state on $6_{even}$ is a~Dirac state (see  \cite{DeSimoneNavaraPtak1}). The~orthomodular poset $4_{even}$ defined analogously on $\{1,2,3,4\}$ can be viewed as a substructure of $6_{even}$ but $4_{even}$ has a two-valued state that is not a~Dirac state (it suffices to set $s(1,2)=s(1,3)=s(1,4)=0$).

In the~last remark, let us observe that even for the~orthomodular posets of $\mathcal{SOMP}$ that are closed under symmetric difference there can be established a link of \textit{Theorem 2.5(iii)} with a generalized Stone representation. It can be done in the~analogy to the~link of \textit{Theorem 2.5(i)} to \textit{Theorem 3.1}. It suffices to consider the~two-valued $\triangle$-states instead of the~mere two-valued states (see also~\cite{DeSimoneNavaraPtak1}; the~state is said to be a $\triangle$-state if $s(A \triangle B) \leq s(A) + s(B), A, B \in \mathcal{L}$). We thus obtain the~class of the~orthomodular posets of $\mathcal{SOMP}$ that are closed under symmetric difference isomorphic to the~class of the~point-distinguishing elements of $\mathcal{SOMP}$ on which all two-valued $\triangle$-states are Dirac.

\section*{Acknowledgement}
The~first author was supported by the~Czech Science Foundation grant 20-09869L.
The~second author was supported by the~European Regional Development Fund, project ``Center for Advanced Applied Science''\\ (No.\ CZ.02.1.01/0.0/0.0/16\_019/0000778).

The final publication is available at Springer via https://doi.org/10.1007/s10773-023-05436-3.

\section*{Affiliations}
\vspace*{10mm}
{\sc Dominika Bure\v{s}ov\'{a}}

Department of Cybernetics

Czech Technical University, Faculty of Electrical Engineering

166 27  Prague 6

Czech Republic

{\em e-mail}: buresdo2@fel.cvut.cz
\\[5mm]
{\sc Pavel Pt\' ak}

Department of Mathematics

Czech Technical University, Faculty of Electrical Engineering

166 27  Prague 6

Czech Republic

{\em e-mail}: ptak@fel.cvut.cz

\section*{Statements and Declarations}
The~first author was supported by the~Czech Science Foundation grant 20-09869L.
The~second author was supported by the~European Regional Development Fund, project ``Center for Advanced Applied Science''\\ (No.\ CZ.02.1.01/0.0/0.0/16\_019/0000778).

The~authors have no relevant financial or non-financial interests to disclose.


\begin{thebibliography}{EKMMW}
	
\bibitem[1]{BikchentaevNavara}
	Bikchentaev, A., Navara, M.:
	States on symmetric logics.
	Mathematica Slovaca \textbf{66}, 359--366 (2016)
	
	
\bibitem[2]{BuresovaPtak}
	Bure\v{s}ov\'{a}, D., Pt\'{a}k, P.: Quantum logics that are symmetric-difference-closed.
	International Journal of Theoretical Physics \textbf{60}, 3919--3926 (2021)
	
\bibitem[3]{DeSimoneNavaraPtak1}
    De Simone, A., Navara, M., Pt\'{a}k, P.:
    Extending states on finite concrete logics.
    International Journal of Theoretical Physics \textbf{46}, 2046--2052 (2007)
    
\bibitem[4]{DeSimoneNavaraPtak2}
    De Simone, A., Navara, M., Pt\'{a}k, P.:
    States on systems of sets that are closed under symmetric difference.
    Mathematische Nachrichten \textbf{288}, no. 17--18, 1995--2000 (2015)
  
\bibitem[5]{DvurecenskijPulmannova}
    Dvure\v{c}enskij, A., Pulmannov\'{a}, S.:
    New Trends in Quantum Structures.
    Kluwer Academic Publishers, Dordrecht
    (2000)

\bibitem[6]{GoGre}
Godowski, R., Greechie, R. J.: Some equations related to states on orthomodular lattices. Demonstratio Mathematica \textbf{17}, 241--250 (1984)

    
\bibitem[7]{Gudder}
    Gudder, S. P.:
    Stochastic Methods in Quantum Mechanics.
    North-Holland, Amsterdam
    (1979)
    
\bibitem[8]{Harding}
    Harding, J.:
    Remarks on concrete orthomodular lattices.
    International Journal of Theoretical Physics \textbf{43}, 2149--2168
    (2004)


\bibitem[9]{HrochPtak}
	Hroch, M., Pt\'{a}k, P:
	Concrete quantum logics, $\triangle$-logics, states and $\triangle$-states. International Journal of Theoretical Physics \textbf{56}, 3852--3859 (2017)

 
\bibitem[10]{Matousek}
	Matou\v{s}ek, M.:
	Orthocomplemented lattices with a symmetric difference. Algebra Universalis \textbf{60}, 185--215 (2009)


\bibitem[11]{Mayet}  
Mayet, R.:
Varieties of orthomodular lattices related to states.
Algebra Universalis \textbf{20}, 368--396 (1985)

\bibitem[12]{MayetPtak}  
Mayet, R., Pt\'{a}k, P:
Orthomodular lattices with state-separated noncompatible pairs.
Czechoslovak Mathematical Journal \textbf{50} (2), 359--366 (2000)


\bibitem[13]{NavaraPtak}
Navara, M., Pt\'{a}k, P:
On Frink ideals in orthomodular posets.
Order \textbf{38}, 245--249 (2021)

	
\bibitem[14]{Ptak}
Pt\'{a}k, P.:
Concrete quantum logics.
 International Journal of Theoretical Physics \textbf{39},
 827--839  (2000)


\bibitem[15]{Tkadlec}
Tkadlec, J.:
Partially additive states on orthomodular posets.
Colloquium Mathematicum \textbf{62}, 7--14 (1991)
   
\end{thebibliography}
\end{document}